\documentclass[12pt,leqno]{article}
\usepackage[margin=1in]{geometry}
\usepackage{amsthm,amsmath,amssymb}
\usepackage{natbib,url,enumerate}
\usepackage{graphicx,xcolor}
\usepackage[colorlinks=true,linkcolor=blue!60!black,citecolor=blue!40!black,urlcolor=blue]{hyperref}
\usepackage{setspace}
\newtheorem{thm}{Theorem}
\newtheorem{prop}[thm]{Proposition} 
\newtheorem{lemma}[thm]{Lemma}

\theoremstyle{definition}

\theoremstyle{remark}

\newtheorem{example}{Example}

\newcommand*{\set}[1]{\left\{#1\right\}} 
\newcommand*{\hi}{H}
\newcommand*{\lo}{L}
\newcommand*{\ty}{k}

\begin{document}
\title{When abstinence increases prevalence} 
\author{Sander Heinsalu\thanks{Research School of Economics, Australian National University. HW Arndt Building, 25a Kingsley St, Acton ACT 2601, Australia.
\newline Email: sander.heinsalu@anu.edu.au, 
website: \url{https://sanderheinsalu.com/}
The author thanks George Mailath, Rakesh Vohra, Steven Matthews and Andrew Postlewaite for comments and suggestions and the University of Pennsylvania for its hospitality during part of this work. 
}}
\date{\today}
\maketitle

\begin{abstract}
In the pool of people seeking partners, a uniformly greater preference for abstinence increases the prevalence of infection and worsens everyone's welfare. In contrast, prevention and treatment reduce prevalence and improve payoffs. 
The results are driven by adverse selection---people who prefer more partners are likelier disease carriers. A given decrease in the number of matches is a smaller proportional reduction for people with many partners, thus increases the fraction of infected in the pool. The greater disease risk further decreases partner-seeking and payoffs. 


	
\end{abstract}


Abstinence education is a widespread policy aiming to combat sexually transmitted diseases (STDs) and unwanted pregnancy. This paper contrasts the consequences of abstinence with preventive measures (inoculation, condoms) and treatment in an economic epidemiology model. Despite seeming similarity, abstinence has the opposite effect to vaccines and cures on disease prevalence and the welfare of people with any preference intensity for an increased number of partners. 

The goal of abstinence programs is to change preferences towards reducing the number of sex partners. If those with fewer partners reduce their number of matches more, then the prevalence of the disease among the active partner-seekers increases. 
The reason 
is adverse selection: those most eager to match are the likeliest to be infected. The abstinence intervention and preference-based behaviour modification are complementary to those considered in the previous literature and imply different welfare consequences. 

A change in prevalence among match-seekers analogous to that induced by partial abstinence of the low-activity people occurs when the disease carriers choose to have more partners. For example, treatment partly restores their health, thus the ability to search for contacts, which was previously reduced by illness. 

Given any preferences, a decrease in the fraction infected cannot increase prevalence in the pool of match-seekers in the short run, even though the optimal number of partners chosen increases in response. The direct effect of the decreased infections is never overwhelmed by the indirect effect of a greater probability of matching, because if the new higher rate of matches raises the weighted prevalence in the pool of partner-seekers to its original level, then the best responses are the same as initially. At the initial numbers of matches, the smaller infected fractions imply a reduced prevalence in the pool. 
By a similar revealed preference reasoning, vaccines always increase payoffs in the short run: the same number of partners as before the vaccine introduction is still a feasible choice, now with a lower chance of infection. 
This contrasts with a change in preferences towards fewer partners, for example due to abstinence counselling, which may increase prevalence and reduce the payoffs of all types, whether evaluated at the old or the new utility function. 
Even after compensating for the direct effect of a reduced benefit from matching, the payoff of all types may fall due to the indirect effect of greater prevalence (riskier partners). 

In addition to abstinence education which specifically targets sexually transmitted diseases, governments try to control other epidemics using `social distancing policies' designed to reduce the population's meeting rate. Examples are public announcements telling people to stay home, closing of schools and mass transit. Exhorting people to avoid public places aims to change preferences, similarly to abstinence education, but school and transit closures try to directly reduce contacts, thus disease transmission, analogously to prevention and treatment. The seemingly similar social distancing policies have opposite effects, just like abstinence and medical interventions. If those with low social activity respond more to public messages, then disease prevalence among people still meeting each other increases, so infections rise. By contrast, decreasing contact rates through closing schools and mass transit reduces disease transmission. 

The literature, discussed next, has mostly focussed on infection rates and medical interventions, but is increasingly considering welfare and incentives as well.

\subsection*{Literature}
\label{sec:lit}

The pathbreaking work of \cite{kremer1996} studied how the chosen matching probabilities and the disease prevalence in the pool of partner-seekers respond to infection rates. The present paper uses \citeauthor{kremer1996}'s two-type model to answer the complementary questions of how abstinence education affects prevalence, the number of partners and the payoffs of people with different preferences for matching. 
The effects of abstinence programs operate via the utility function, and contrast with the effects of prevention and treatment, which work via infection rates. 

Abstinence education modifies preferences, which changes choices. An agent who expects altered actions of others in turn changes his best response, and this indirect influence may strengthen or overturn the original direct effect of modifying the utility function. The current paper derives the optimal number of partners from preferences, accounting for the feedback from the composition of the matching pool to the choices. By contrast, \cite{whitaker+renton1992,kremer1994,kremer+morcom1998} take the frequency of partner change as given, thus assume the absence of the indirect effect of interacting best responses. Reducing the number of partners for people with low activity increases the average probability of infection in the pool of available matches. 
The present paper qualifies this result when the probability of finding a partner is endogenous---a greater preference for abstinence may counterintuitively increase the number of matches (Example~\ref{ex:counterintuitive}). 
This interesting effect has not been examined in the epidemiology literature, to the best of the author's knowledge. 

\cite{heinsalu2017,heinsalu2018} show a similar counterintuitive comparative static for informed sellers investing to gain (earlier) access to a market of uninformed buyers---lower costs may reduce trade and profits. Investment to enter the market is similar to the cost of finding matches (or satiation in the number of partners) in the present work. Expected quality in the market is analogous to the negative of the infection prevalence in the partner pool. Reducing the production cost of the seller types resembles increasing the desire for more matches. 

In \cite{fenichel2013}, agents' preferences depend on their health status, which may change endogenously from susceptible to infected and from infected to immune (no other transitions are possible). Infected and immune individuals do not behave strategically, unlike in the current work where all types optimise. \citeauthor{fenichel2013} studies social distancing policies for a benevolent planner who chooses agents' behaviour, with or without the constraint of identical action change for all health statuses.
This constraint differs from restricting the planner to provide to identical incentives to reduce contacts (one of the policies considered in the current paper), because types with unequal utility functions respond differently to the same incentives.
The present paper complements \cite{fenichel2013} by comparing policies that change preferences to those cutting infection rates. For disease prevalence, reducing infections is equivalent to decreasing the number of contacts directly, but the welfare effects are the opposite. Involuntary behaviour change tends to decrease utility, but a lower risk of infection improves welfare. 

In \cite{greenwood+2017}, in each period people search (at a cost) first in the marriage market, then in markets for unprotected and protected sex. Married agents' sex is assumed unprotected, but monogamous, in contrast to the present paper where all agents can search for partners and, conditional on infection, have the same transmission probability. \cite{greenwood+2017} focus on long run steady states and assume everyone knows his or her infection status, unlike in the current work. In numerical simulations, increasing the utility of marriage (e.g.\ by abstinence education) or reducing exogenous divorces decreases disease prevalence. By contrast, the present paper shows analytically when abstinence increases prevalence and reduces welfare. 

\cite{talamas+vohra2019} focus on vaccination, not abstinence, in a commonly known social network where homogeneous agents choose who to partner with. In the current work, heterogeneous people match randomly. In \cite{talamas+vohra2019}, partially effective vaccines may make the best stable network more connected, increasing disease transmission and worsening welfare---the opposite of the current paper. Risk choices in their network are strategic complements: an agent with more promiscuous partners has a lower marginal cost of adding risky interactions, because of a greater probability of being already infected. In the present work, contact frequencies are strategic substitutes for high-risk types, but complements for low-risk. Decreasing the benefit from matching in \cite{talamas+vohra2019} would reduce stable network connectivity and disease prevalence---the opposite outcome to the present work. 

Other theoretical articles modelling infection rates and immunisation include \cite{galeotti+rogers2013,chen+toxvaerd2014,rowthorn+toxvaerd2017,goyal+vigier2015,goyal+2016,toxvaerd2019}. The current paper complements this literature by considering welfare and disease prevalence under both abstinence programs and vaccination.

Empirically, \cite{underhill+2007} find that among 13 trials of abstinence programs for a total of 15 940 US youths, most had no effect, one reduced vaginal sex in the short term and one increased the frequency of STDs in both the short and the long term. The current paper theoretically explains how abstinence education that influences low-risk people more increases disease prevalence. 
In the data, those with less sexual experience are affected more by abstinence education. For example, \cite{jemmott+1998} find that abstinence programs reduce the frequency of sex more among those with a lower baseline. The reduction is statistically significant after 3 months, but not after 6 or 12. 
The abstinence-only intervention in \cite{jemmott+2010} delayed first-time sex, but did not influence the average number of partners, condom use or other outcomes. 
\cite{bruckner+bearman2005} find that those making an abstinence pledge delay first-time sex, but have insignificantly higher STD rates. 
In \cite{odonnell+1999}, at 6- and 12-month follow-up after safer-sex interventions, the reduction in the probability of intercourse in the previous 3 months is proportionally larger (but in absolute terms, smaller) for adolescents who are virgins at baseline. 

Analogously to an increased preference for abstinence, bad health may lead to a lower desire for sex. 
\cite{lakdawalla+2006} empirically show that advances in treatment caused better health for HIV carriers, which increased their sexual activity and rekindled HIV spread. 

Social disapproval of (some varieties of) sex has an effect similar to a preference for abstinence. 
\cite{francis+mialon2010} find that tolerance for sexual minorities is negatively associated with the HIV rate, leads low-risk gay men to enter the pool partners, and more active men to choose less risky behaviours. 
\cite{auld1996,auld2006} show that during the 1982--1984 AIDS outbreak in San Francisco, the average number of sex partners for gay men fell by a third, but the standard deviation and skewness increased. Thus high-activity individuals responded less to an increase in the perceived infection probability, as in the model below.

\section{Search for partners in an adversely selected pool}
\label{sec:model}

The model and notation follow \cite{kremer1996}. 
Each person has a type $\ty\in\set{\hi,\lo}$ (interpret $H$ as high-risk, $L$ as low), with fraction $\alpha_{\ty}$ being type $\ty$. Type is private information, but the fractions are commonly known. Each type chooses the probability $i_{\ty}\in[0,1]$ of matching with a partner per unit of time.\footnote{
By re-normalising the unit of time, choosing the probability of a match is equivalent to choosing the number of partners per unit of time. 
} 
The fraction of type $\ty$ people who are infected is denoted $Y_{\ty}$, with $Y_{\hi}>Y_{\lo}\geq 0$. The probability that a randomly chosen partner carries the disease is then 
\begin{align}
\label{W}
W:=\frac{\alpha_{\hi}i_{\hi}Y_{\hi}+\alpha_{\lo}i_{\lo}Y_{\lo}}{\alpha_{\hi}i_{\hi}+\alpha_{\lo}i_{\lo}}. 
\end{align} 
Normalise the cost of infection times the probability of catching the disease from an infected contact to $1$ w.l.o.g., so the marginal cost of matching with a random partner is $W$ for an uninfected individual. 
The payoff of type $\ty$ is 
$\pi_{\ty}(i_{\ty}) =i_{\ty}\theta_{\ty}-\psi i_{\ty}^2-i_{\ty}W 
$, where $\psi\geq \frac{\theta_{\hi}-Y_{\lo}}{2}$\footnote{
The assumption $\psi\geq \frac{\theta_{\hi}-Y_{\lo}}{2}$ is sufficient for $i_{\ty}<1$ (interior solution) and is implicit in \cite{kremer1996} Eq.~(4), which otherwise becomes $i_{\ty}=\min\set{i_{\max},\max\set{\frac{\theta_{\ty}-W}{2\psi},0}}$. 
} 
and $\theta_{\hi}>\theta_{\lo}>0$. 
The marginal benefit $\theta_{\ty}-2\psi i_{\ty}$ of matching with more partners strictly decreases in the number of partners, interpreted as a satiation effect or as an increasing cost of effort of finding matches. 

Interpret changes in behaviour $i_{\ty}$ at constant infection rates $Y_{\ty}$ as occurring in the short run. Feedback from behaviour to the fraction carrying the disease ($Y_{\ty}$ changing in response to $i_{\ty}$) is interpreted as taking place in the long run. For example, the susceptible-infected (SI) models in epidemiology focus on the long run. Feedback from the chosen matching probabilities $i_{\ty}$ to the composition $W$ of the pool of active partner-seekers naturally occurs in the short run (the mix of people found in a singles bar changes as soon as behaviour does), and this is the focus from now on. 

Equilibrium consists of $i_{\hi}^*$, $i_{\lo}^*$ and $W^*$ s.t.\ 
$i_{\ty}^*$ maximises $\pi_{\ty}(i_{\ty})$ given $W^*$, which in turn is derived from $i_{\ty}^*$ using~(\ref{W}). 
Based on $\pi_{\ty}(i_{\ty})$, the best responses to any $W$ satisfy $i_{\hi}\geq i_{\lo}$, and if $i_{\lo}>0$, then $i_{\hi}>i_{\lo}$. Equilibria in which $i_{\lo}=0$ are called \emph{$\lo$-inactive} and those in which $i_{\lo}>0$ \emph{both-active}. 
The next section studies the existence and properties of equilibria, followed by the comparative statics that constitute the main result.

\section{The effects of abstinence and treatment on search and welfare}
\label{sec:results}

This section first establishes the existence of equilibria and their regions of uniqueness. It then proceeds to analyse the stable equilibria, often focussing on the Pareto dominant one. 

The following result (Proposition I.~of \cite{kremer1996}) describes when different equilibrium combinations exist. 
\begin{prop}
\label{prop:kremer}
1) If $Y_{\hi}<\theta_{\lo}$, then a unique equilibrium exists, with $i_{\lo}>0$. 
\\ 2) If $\theta_{\lo} <Y_{\hi} <\theta_{\lo}-\theta_{\hi}+[\theta_{\lo}-\alpha_{\lo}Y_{\lo}]/\alpha_{\hi}$ and $[\theta_{\lo}-\alpha_{\hi}(\theta_{\hi}-\theta_{\lo})-\alpha_{\lo}Y_{\lo}-\alpha_{\hi}Y_{\hi}]^2-4\alpha_{\hi}(\theta_{\hi}-\theta_{\lo})(\theta_{\lo}-Y_{\hi})>0$, then there exists one $\lo$-inactive and two both-active equilibria. 
\\ 3) If 
$\theta_{\lo} <Y_{\hi} >\theta_{\lo}-\theta_{\hi}+[\theta_{\lo}-\alpha_{\lo}Y_{\lo}]/\alpha_{\hi}$, then a unique equilibrium exists, with $i_{\lo}=0$. 
\end{prop}

Equilibria are Pareto ranked in reverse order of weighted prevalence $W^*$, both for a given set of parameters and, fixing $\psi$, $\theta_{\hi}$, across parameters, as shown next. 
\begin{lemma}
\label{lem:monot}
Fix $\psi$, then any equilibrium with a higher $i_{\ty}^*$ has a higher $\pi_{\ty}(i_{\ty}^*)$ than an equilibrium with a lower $i_{\ty}^*$. 
Fix $\psi$, $\theta_{\hi}$, then any equilibrium with a lower $W^*$ has a higher $i_{\ty}^*$. 
\end{lemma}
\begin{proof}
%
Substitute the FOC 
$i_{\ty}^*=\max\set{0,\frac{\theta_{\ty}-W^*}{2\psi}}$ into $\pi_{\ty}$, which becomes 
$
\psi i_{\ty}^{*2}$, with $\frac{d\pi_{\ty}}{di_{\ty}^*}>0$. 
Changing parameters other than $\theta_{\ty}$ only affects $i_{\ty}^*$ via $W^*$, and $\frac{di_{\ty}^*}{dW^*}<0$. Due to $\frac{\partial i_{\lo}^*}{\partial \theta_{\lo}}>0$, $\frac{\partial i_{\lo}^*}{\partial W^*}<0$ and $\frac{\partial W^*}{\partial \theta_{\lo}}<0$, an equilibrium with a lower $W^*$ has a higher $i_{\lo}^*$. 
\end{proof}
Intuitively, both types reduce activity in response to greater disease incidence, and by revealed preference become worse off, because the match probability chosen at high prevalence could be chosen at low and would yield greater utility than at high. The optimal match probability at a low infection rate leads to an even larger payoff. 

It is natural for people to coordinate on the Pareto dominant equilibrium, which by Lemma~\ref{lem:monot} has the greatest probability of matching and the lowest prevalence of infection. 
With multiple equilibria, the one with the second-highest $i_{\ty}^*$ is unstable, as shown in Lemma~\ref{lem:unstable} in the Appendix. Focus on stable equilibria from now on. 

The following main result examines how equilibria change in response to abstinence, prevention or treatment. 
Treatment is interpreted as directly reducing infection rates $Y_{\ty}$ among both types, possibly by different amounts. 
The cost of infection also falls when treatment is available, which is equivalent to a proportional fall in both $Y_{\ty}$. Preventative measures decrease the probability of catching the disease given $Y_{\ty}$, which has the same effect as a proportional reduction in $Y_{\ty}$. Thus the direct effects of preventions and cures may both be modelled as decreasing $Y_{\hi}$ and $Y_{\lo}$. 
The indirect effect of treatment is that both types increase their probability of matching in response to a lower perceived risk of infection. 
As shown next in Proposition~\ref{prop:compstats}, this indirect effect cannot outweigh the direct in the short run.\footnote{In the long run, feedback from the matching probabilities $i_{\ty}$ to the fractions infected $Y_{\ty}$ may overwhelm the initial direct effect of lower $Y_{\ty}$ and thereby increase the disease rate \citep{kremer1994}.
} 
Thus preventions and cures always cut the weighted prevalence in the pool of active partner-seekers. The welfare of all types increases as a result (Lemma~\ref{lem:monot}: payoffs increase in the matching probabilities). By contrast, changing preferences towards abstinence (reducing $\theta_{\ty}$) may increase the fraction of infected in the pool and decrease matching and the payoffs under both the old and the new utility function. 

\begin{prop}
\label{prop:compstats}
In any stable equilibrium, 
$\frac{dW^*}{dY_{\ty}}\geq 0$, $\frac{d i_{\ty}^*}{d Y_{\hat{\ty}}}\leq 0$ for $\ty,\hat{\ty}\in\set{\hi,\lo}$, and $\frac{dW^*}{d\theta_{\lo}}\leq 0\leq \frac{dW^*}{d\theta_{\hi}}$, $\frac{d i_{\ty}^*}{d \theta_{\lo}}\geq 0$, $\frac{d i_{\lo}^*}{d \theta_{\hi}}\leq 0$, 
and increasing $\theta_{\hi}$ less than $\theta_{\lo}$ reduces $W^*$ and increases $i_{\ty}^*$. 
In any $\lo$-inactive equilibrium, $\frac{di_{\hi}^*}{d\theta_{\hi}} \geq 0 =\frac{dW^*}{d\theta_{\hi}}$. 
If a decrease in $\theta_{\hi}$, 
$Y_{\ty}$ or an increase in $\theta_{\lo}$ 
changes the equilibrium set, then the stable equilibrium with the greatest $W^*$ disappears or a stable equilibrium with a new smallest $W^*$ appears. 
\end{prop}
The proof of Proposition~\ref{prop:compstats} is in the Appendix. 
The proposition shows that reducing $\theta_{\lo}$ weakly more than $\theta_{\hi}$ (the empirically observed change) increases weighted prevalence in the pool and reduces both types' activity and payoffs. 
The direct effect of abstinence education is a lower benefit $\theta_{\ty 1}\leq \theta_{\ty 0}$ from match-seeking, which for any prevalence $W$ and search intensity $i_{\ty}$ reduces the payoff, as evaluated at the new utility function $\pi_{\ty 1}(i_{\ty})=(\theta_{\ty 1}-W)i_{\ty}-\psi i_{\ty}^2$. 
The indirect effect of the changing composition of the pool of partners strengthens the direct effect when the low types reduce their activity more, because then $W$ increases. 
At the old utility function $\pi_{\ty 0}$, the new optimal activity level $i_{\ty1}^*$ was available but not chosen, thus by revealed preference payoffs decrease when evaluated at the old utility function: 
$\pi_{\ty 0}(i_{\ty 0}^*)\geq \pi_{\ty 0}(i_{\ty 1}^*)\geq \pi_{\ty 1}(i_{\ty 1}^*)$. 

Even if both types are compensated for the direct decrease in their benefit, their payoffs may fall (Example~\ref{ex:compensate} in the Appendix). Compensation means adding the benefit change times the initial activity level $(\theta_{\ty0}-\theta_{\ty1})i_{\ty0}^*$ to the payoff $\pi_{\ty 1}$ after the change.\footnote{Reimbursing agents for the utility loss at their original choice is analogous to a Slutsky compensated price change in consumer theory. In a single agent decision, by revealed preference, payoff cannot decrease after Slutsky compensation. In the game in this paper, the externality of rising disease prevalence in response to lower benefits from matching may be sufficiently strong to outweigh the compensation.} 
In reality, the compensation may be interpreted as an increased utility from alternative activities in response to abstinence education that reduces enjoyment of casual sex. This `preference substitution' is unlikely to be large enough in practice to keep people's payoff constant when they do not change their activity level, so the compensation considered above should be thought of as an upper bound. 


In Proposition~\ref{prop:compstats}, a sufficient condition for prevalence to fall and the probabilities of matching to increase after intensifying the preference for a greater number of partners is that the benefit from more matches increases less for high-activity people. This condition may be weakened, but cannot be omitted entirely, as shown next.  
Raising $\theta_{\hi}$ always increases prevalence and may reduce both types' matching probabilities and payoffs, as the following example demonstrates.  
\begin{example}
\label{ex:counterintuitive}
Take $\psi=0.1$, $\alpha_{\ty}=0.5$, $Y_{\lo}=0.1$, $Y_{\hi}=0.6$, $\theta_{\lo}=0.45$. If $\theta_{\hi0}=0.46$, then in the unique stable equilibrium, $(i_{\lo0}^*,i_{\hi0}^*)\approx (0.43,0.48)$ and $W^*_{0}\approx 0.38$, but if $\theta_{\hi1}=0.47$, then $(i_{\lo1}^*,i_{\hi1}^*)\approx (0.34,0.44)$ and $W^*_{1}\approx 0.36$. 
\end{example}
Adding a small enough increase in $\theta_{\lo}$ to Example~\ref{ex:counterintuitive} does not overturn the qualitative result that a greater benefit from finding a partner increases disease prevalence in the pool of match-seekers and reduces the matching rates and payoffs. 
Reducing people's benefit $\theta_{\ty}i_{\ty}-\psi i_{\ty}^2$ from an interaction directly decreases their payoff, but Example~\ref{ex:counterintuitive} shows that the indirect effect of the changing composition of the pool of partners may overturn the direct effect. 
This counterintuitive comparative static illustrates that taking preferences into account may yield diametrically different predictions from assuming fixed behaviour (e.g.\ the exogenous number of partners in \cite{kremer1994}). 

To complete the comparative statics in Proposition~\ref{prop:compstats}, the following argument establishes another interesting prediction of \cite{kremer1996}: an increase in the matching cost $\psi$ (faster satiation) does not affect prevalence in any equilibrium nor change the set of equilibria. In $\lo$-inactive equilibria, $W^*=Y_{\hi}$ and $i_{\lo}^*$ are constant in $\psi$, and $i_{\hi}^*$ remains positive for all $\psi$. In both-active equilibria, with linear marginal cost, positive matching probabilities $i_{\ty}^*$ fall proportionately in $\psi$, which leaves $W^*$ unchanged, thus has no indirect effect or feedback to $i_{\ty}^*$. Payoffs and matching probabilities of course decrease in $\psi$ in any equilibrium. 

Arguments in \cite{kremer1996} often rely on equiproportionate reductions in activity by the types, for example in response to rising weighted prevalence in the pool. To the author's knowledge, the literature has not pointed out that increasing the satiation parameter $\psi$ causes such proportional reductions. 
It is worth mentioning that for both-active equilibria, the result relies on the quadratic cost of search. However, the arguments using proportionally decreasing matching probabilities still hold when low types reduce their activity by a greater factor than high, which is the case for a wide class of search costs (satiation functions). 

Having established the formal results, the next section proceeds to discuss extensions and implications of the model.

\section{Discussion}
\label{sec:discussion}

The results in the preceding section generalise to a wide variety of strictly convex costs of increasing the matching probability, as may be seen in the similar framework of \cite{heinsalu2018}. That paper also shows that the results change qualitatively when the cost is close to linear, because then the matching probability is always either minimal or maximal. 

The qualitative results remain unchanged when two populations $g\in\set{f,m}$ match with each other, for example females and males. 
The weighted infection rate in population $g$ is $W_{g}:=\frac{\alpha_{g\hi}i_{g\hi}Y_{g\hi}+\alpha_{g\lo}i_{g\lo}Y_{g\lo}}{\alpha_{g\hi}i_{g\hi}+\alpha_{g\lo}i_{g\lo}}$, which is the prevalence in the pool of partners for population $\hat{g}\neq g$. 
Male-to-female HIV transmission rate is higher than the female-to-male rate, which can be captured by $\beta>1$ in the utility function $\pi_{f\ty}(i_{f\ty}) =(\theta_{f\ty}-\beta W_{m})i_{f\ty}-\psi i_{f\ty}^2$. The factor $\beta$ also describes different costs of infection between the populations. For example, the human papilloma virus (HPV) harms females more, which corresponds to a greater $\beta$. 

Type $\ty$ in population $g$ responds to own preferences $\theta_{g\ty}$ and the infection rate $W_{\hat{g}}$ in the other population. If both populations are subject to similar prevention, treatment or abstinence interventions, then the effects resemble the single population case discussed above. 

Embedding the short run framework of this paper in an epidemiological SI model is straightforward and allows the infection rates of the types to respond the probability of matching and the weighted prevalence in the pool of partners. In the long run, the additional feedback from infection rates to behaviour may change the comparative statics of vaccination and treatment, as the previous literature shows. The current paper's results from abstinence-caused preference change are reinforced in the long run when higher weighted prevalence in the partner pool increases infection rates for both types. 

As this paper shows, abstinence programs and other preference-altering social distancing policies may backfire during epidemics, not just in terms of disease prevalence, but also by reducing welfare. On the other hand, medical interventions unambiguously cut infections and improve payoffs, at least in the short run. By the principle of `first, do no harm', priority should be given to policies directly targeting disease transmission, unless compelling countervailing reasons exist.

\appendix
\section{Appendix}
\label{sec:appendix}

\begin{proof}[Proof of Proposition~\ref{prop:compstats}]
In any equilibrium, by definition $W^*\in[Y_{\lo},Y_{\hi}]$. Best responses to any $W^*$ satisfy $i_{\lo}^*\leq i_{\hi}^*$, so $W^*\geq \alpha_{\hi}Y_{\hi}+\alpha_{\lo}Y_{\lo}$. If $\theta_{\hi}\leq W^*$, then $i_{\ty}^*= \pi_{\ty}^*=0$ and $W^*=Y_{\hi}$ (worst off-path belief), so the comparative statics are trivial: neither $i_{\ty}^*$ nor $\pi_{\ty}^*$ responds to parameters, and $W^*$ responds only to $Y_{\hi}$. 
Suppose $W^*<\theta_{\hi}$ from now on. 

Substitute the BR $i_{\ty}^*=\max\set{0,\frac{\theta_{\ty}-W^*}{2\psi}}$ into~(\ref{W}) to obtain 
\begin{align}
\label{wfunc}
w(W^*) 
=\frac{\alpha_{\hi}Y_{\hi}+\alpha_{\lo}Y_{\lo}\max\set{0,\theta_{\lo}-W^*}/(\theta_{\hi}-W^*)}{\alpha_{\hi}+\alpha_{\lo}\max\set{0,\theta_{\lo}-W^*}/(\theta_{\hi}-W^*)}.
\end{align} 
Equilibria are fixed points of $w(\cdot)$, because both types are best responding to the $W^*$ they expect and the resulting $w(W^*)=W^*$ justifies the expectation. 
If $W^*\geq\theta_{\lo}$, then $w'(W^*)=0$, and if $W^*<\theta_{\lo}$, then $w'(W^*)>0$, because $\frac{d}{dW^*}\left(\frac{\theta_{\lo}-W^*}{\theta_{\hi}-W^*}\right) =\frac{-\theta_{\hi}+W^*+\theta_{\lo}-W^*}{(\theta_{\hi}-W^*)^2}<0$. An equilibrium is stable iff $|w'(W^*)|<1$. 

To simplify the study of the fixed points of $w(\cdot)$, define 
\begin{align*}
\delta(W^*):&=[W^*-w(W^*)][\alpha_{\hi}(\theta_{\hi}-W^*)+\alpha_{\lo}\max\set{0,\theta_{\lo}-W^*}] 
\\&\notag=\alpha_{\hi}(\theta_{\hi}-W^*)(W^*-Y_{\hi})+\alpha_{\lo}\max\set{0,\theta_{\lo}-W^*}(W^*-Y_{\lo}).
\end{align*} 
The zeroes of $\delta(W^*)$ are the fixed points of $w(W^*)$, because $\theta_{\hi}> W^*$, so $\alpha_{\hi}(\theta_{\hi}-W^*)>0$. 
Clearly $\delta(Y_{\lo})<0\leq \delta(Y_{\hi})$. By continuity of $\delta$, equilibrium exists. 
An equilibrium with prevalence $W^*$ is stable iff $\delta'(W^*)>0$. 
Continuity of $\delta$ implies that stable and unstable equilibria alternate.

The function $\delta(\cdot)$ pointwise increases in $\theta_{\lo}$, $\alpha_{\lo}$ and decreases in $\theta_{\hi}$, $\alpha_{\hi}$, $Y_{\ty}$. 
Increasing a function pointwise shifts the zeroes at which the function increases left and the zeroes at which the function decreases right \citep{milgrom+roberts1994}. 
Thus in stable equilibria, $W^*$ decreases in $\theta_{\lo}$, $\alpha_{\lo}$ and increases in $\theta_{\hi}$, $\alpha_{\hi}$, $Y_{\ty}$ (note $W^*$ differs from the $W$ defined in~(\ref{W}) and the function $w$ in~(\ref{wfunc})). 
The partial derivatives $\frac{\partial i_{\hi}^*}{\partial W^*}< 0$, $\frac{\partial i_{\lo}^*}{\partial W^*}\leq 0$, $\frac{\partial i_{\hi}^*}{\partial \psi}< 0$, $\frac{\partial i_{\lo}^*}{\partial \psi}\leq 0$, $\frac{\partial i_{\hat{\ty}}^*}{\partial Y_{\ty}}=0$, $\frac{\partial i_{\hi}^*}{\partial \theta_{\lo}} =\frac{\partial i_{\lo}^*}{\partial \theta_{\hi}}=0$, $\frac{\partial i_{\ty}^*}{\partial \theta_{\ty}}\geq 0$ 
imply the total derivatives $\frac{d i_{\ty}^*}{d Y_{\hat{\ty}}}\leq 0$ for $\ty,\hat{\ty}\in\set{\hi,\lo}$, and $\frac{d i_{\ty}^*}{d \theta_{\lo}}\geq 0$, $\frac{d i_{\lo}^*}{d \theta_{\hi}}\leq 0$. 

Define 
$
\delta(W^*,\gamma) :=\alpha_{\hi}\gamma(1-\epsilon)(\theta_{\hi}-W^*)(W^*-Y_{\hi})+\alpha_{\lo}\gamma\max\set{0,\theta_{\lo}-W^*}(W^*-Y_{\lo})
$
for $\epsilon\geq0$, $\gamma>0$. Then $\frac{d\delta(W^*,\gamma)}{d \gamma} 
=\delta(W^*,1)-\epsilon\alpha_{\hi}(\theta_{\hi}-W^*)(W^*-Y_{\hi}) >\delta(W^*)$, so at any $W_0^*$ at which $\delta(W_0^*)=0$, we have $\frac{d\delta(W_0^*,\gamma)}{d \gamma}>0$. At stable equilibria therefore $\frac{d W^*}{d \gamma}<0$ and $\frac{d i_{\ty}^*}{d \gamma}>0$. 
Due to $\theta_{\hi}-W^*>\max\set{0,\theta_{\lo}-W^*}$, increasing both $\theta_{\ty}$ by the same amount (or $\theta_{\hi}$ by less than $\theta_{\lo}$) does not increase $\frac{\theta_{\hi}-W^*}{\theta_{\lo}-W^*}$, thus reduces $W^*$ and increases $i_{\ty}^*$. 
In $\lo$-inactive, $W^*=Y_{\hi}$, so $\frac{di_{\hi}^*}{d\theta_{\hi}}>0$. 

At its smallest and largest zero, $\delta$ increases, because $\delta(Y_{\lo})<0\leq \delta(Y_{\hi})$. 
A pointwise increase in a function that increases at its smallest and largest zero either leaves the set of zeroes unchanged, removes zeroes above the smallest or adds new zeroes below the smallest \citep{milgrom+roberts1994}. By Proposition~\ref{prop:kremer} from \cite{kremer1996}, $\delta$ has at most 3 zeroes. Therefore the set of stable equilibria $\set{W^*:\delta'(W^*)>0}$ shifts down unambiguously when $\delta$ increases pointwise. 
\end{proof}

\begin{lemma}
\label{lem:unstable}
The equilibrium with the second-lowest $W^*$ has $\frac{dW}{dW^*}>1$. 
\end{lemma}
\begin{proof}
If $\theta_{\hi}\leq W^*$, then the unique equilibrium is $i_{\ty}^*$ and the lemma holds vacuously. Suppose $\theta_{\hi}> W^*$ from now on. 
Substitute $i_{\ty}^*$ into~(\ref{W}) to obtain 
$W(W^*) =\frac{\alpha_{\hi}(\theta_{\hi}-W^*)Y_{\hi}+\alpha_{\lo}\max\set{0,\theta_{\lo}-W^*}Y_{\lo}}{\alpha_{\hi}(\theta_{\hi}-W^*)+\alpha_{\lo}\max\set{0,\theta_{\lo}-W^*}}$, which is continuous. Due to $\theta_{\hi}> W^*$, the $\max$ in $i_{\hi}^*$ may be omitted w.l.o.g. 
Equilibria are fixed points $W(W^*)=W^*$. Due to $\lim_{W^*\rightarrow-\infty}W(W^*)=\alpha_{\hi}Y_{\hi}+\alpha_{\lo}Y_{\lo}>W^*$, the smallest fixed point $W^*_{\min}$ features $\frac{dW}{dW^*}|_{W=W^*_{\min}}<1$. 
Thus for $\epsilon>0$ small enough, $W(W^*_{\min}+\epsilon)<W^*_{\min}+\epsilon$. For the second-smallest fixed point $W^*_2$ then $W(W^*_{2}-\epsilon)>W^*_{2}-\epsilon$ and $\frac{dW}{dW^*}|_{W=W^*_{2}}>1$. 
\end{proof}

\begin{example} 
\label{ex:compensate}
Let $\psi=0.1$, $\alpha_{\ty}=0.5$, $Y_{\lo}=0.1$, $Y_{\hi}=\theta_{\lo0}=0.3$, $\theta_{\hi0}=0.31$, then the chosen matching probabilities are $i_{\lo0}^*=0.475$, $i_{\hi0}^*=0.525$, and the payoffs are $\pi_{\lo 0}^*\approx 0.023$, $\pi_{\hi 0}^*\approx 0.028$. Reducing the benefit of activity to $\theta_{\lo1}=0.27$, $\theta_{\hi1}=0.3$ results in $i_{\lo1}^*=0.225$, $i_{\hi1}^*=0.375$, $\pi_{\lo 1}^*\approx 0.005$, $\pi_{\hi 1}^*\approx 0.014$, so Slutsky compensating the types by $(\theta_{\lo0}-\theta_{\lo1})i_{\lo0}^*\approx 0.014$, $(\theta_{\hi0}-\theta_{\hi1})i_{\hi0}^*\approx 0.005$ is insufficient to raise the payoffs to the initial level. 
\end{example}

\bibliographystyle{ecta}
\bibliography{teooriaPaberid} 

\begin{thebibliography}{26}
\newcommand{\enquote}[1]{``#1''}
\expandafter\ifx\csname natexlab\endcsname\relax\def\natexlab#1{#1}\fi

\bibitem[\protect\citeauthoryear{Auld}{Auld}{1996}]{auld1996}
\textsc{Auld, M.~C.} (1996): \enquote{Choices, beliefs, and infectious disease
  dynamics,} Working paper 938, Queen's University Department of Economics.

\bibitem[\protect\citeauthoryear{Auld}{Auld}{2006}]{auld2006}
---\hspace{-.1pt}---\hspace{-.1pt}--- (2006): \enquote{Estimating behavioral
  response to the AIDS epidemic,} \emph{The BE Journal of Economic Analysis \&
  Policy}, 5.

\bibitem[\protect\citeauthoryear{Br{\"u}ckner and Bearman}{Br{\"u}ckner and
  Bearman}{2005}]{bruckner+bearman2005}
\textsc{Br{\"u}ckner, H. and P.~Bearman} (2005): \enquote{After the promise:
  the STD consequences of adolescent virginity pledges,} \emph{Journal of
  Adolescent Health}, 36, 271--278.

\bibitem[\protect\citeauthoryear{Chen and Toxvaerd}{Chen and
  Toxvaerd}{2014}]{chen+toxvaerd2014}
\textsc{Chen, F. and F.~Toxvaerd} (2014): \enquote{The economics of
  vaccination,} \emph{Journal of theoretical biology}, 363, 105--117.

\bibitem[\protect\citeauthoryear{Fenichel}{Fenichel}{2013}]{fenichel2013}
\textsc{Fenichel, E.~P.} (2013): \enquote{Economic considerations for social
  distancing and behavioral based policies during an epidemic,} \emph{Journal
  of health economics}, 32, 440--451.

\bibitem[\protect\citeauthoryear{Francis and Mialon}{Francis and
  Mialon}{2010}]{francis+mialon2010}
\textsc{Francis, A.~M. and H.~M. Mialon} (2010): \enquote{Tolerance and HIV,}
  \emph{Journal of Health Economics}, 29, 250--267.

\bibitem[\protect\citeauthoryear{Galeotti and Rogers}{Galeotti and
  Rogers}{2013}]{galeotti+rogers2013}
\textsc{Galeotti, A. and B.~W. Rogers} (2013): \enquote{Strategic immunization
  and group structure,} \emph{American Economic Journal: Microeconomics}, 5,
  1--32.

\bibitem[\protect\citeauthoryear{Goyal, Jabbari, Kearns, Khanna, and
  Morgenstern}{Goyal et~al.}{2016}]{goyal+2016}
\textsc{Goyal, S., S.~Jabbari, M.~Kearns, S.~Khanna, and J.~Morgenstern}
  (2016): \enquote{Strategic network formation with attack and immunization,}
  in \emph{International Conference on Web and Internet Economics}, Springer,
  429--443.

\bibitem[\protect\citeauthoryear{Goyal and Vigier}{Goyal and
  Vigier}{2015}]{goyal+vigier2015}
\textsc{Goyal, S. and A.~Vigier} (2015): \enquote{Interaction, protection and
  epidemics,} \emph{Journal of Public Economics}, 125, 64--69.

\bibitem[\protect\citeauthoryear{Greenwood, Kircher, Santos, and
  Tertilt}{Greenwood et~al.}{2017}]{greenwood+2017}
\textsc{Greenwood, J., P.~Kircher, C.~Santos, and M.~Tertilt} (2017):
  \enquote{The role of marriage in fighting HIV: a quantitative illustration
  for Malawi,} \emph{American Economic Review}, 107, 158--62.

\bibitem[\protect\citeauthoryear{Heinsalu}{Heinsalu}{2017}]{heinsalu2017}
\textsc{Heinsalu, S.} (2017): \enquote{Investing in trading opportunities in
  dynamic adverse selection,} Working paper.

\bibitem[\protect\citeauthoryear{Heinsalu}{Heinsalu}{2018}]{heinsalu2018}
---\hspace{-.1pt}---\hspace{-.1pt}--- (2018): \enquote{Investing to access an
  adverse selection market,} Working paper.

\bibitem[\protect\citeauthoryear{Jemmott, Jemmott, and Fong}{Jemmott
  et~al.}{2010}]{jemmott+2010}
\textsc{Jemmott, J.~B., L.~S. Jemmott, and G.~T. Fong} (2010):
  \enquote{Efficacy of a theory-based abstinence-only intervention over 24
  months: a randomized controlled trial with young adolescents,} \emph{Archives
  of pediatrics \& adolescent medicine}, 164, 152--159.

\bibitem[\protect\citeauthoryear{Jemmott~III, Jemmott, and Fong}{Jemmott~III
  et~al.}{1998}]{jemmott+1998}
\textsc{Jemmott~III, J.~B., L.~S. Jemmott, and G.~T. Fong} (1998):
  \enquote{Abstinence and safer sex HIV risk-reduction interventions for
  African American adolescents: a randomized controlled trial,} \emph{Journal
  of the American Medical Association}, 279, 1529--1536.

\bibitem[\protect\citeauthoryear{Kremer}{Kremer}{1994}]{kremer1994}
\textsc{Kremer, M.} (1994): \enquote{Can Having Fewer Partners Increase
  Prevalence of Aids?} Working Paper 4942, National Bureau of Economic
  Research.

\bibitem[\protect\citeauthoryear{Kremer}{Kremer}{1996}]{kremer1996}
---\hspace{-.1pt}---\hspace{-.1pt}--- (1996): \enquote{Integrating behavioral
  choice into epidemiological models of AIDS,} \emph{The Quarterly Journal of
  Economics}, 111, 549--573.

\bibitem[\protect\citeauthoryear{Kremer and Morcom}{Kremer and
  Morcom}{1998}]{kremer+morcom1998}
\textsc{Kremer, M. and C.~Morcom} (1998): \enquote{The effect of changing
  sexual activity on HIV prevalence,} \emph{Mathematical biosciences}, 151,
  99--122.

\bibitem[\protect\citeauthoryear{Lakdawalla, Sood, and Goldman}{Lakdawalla
  et~al.}{2006}]{lakdawalla+2006}
\textsc{Lakdawalla, D., N.~Sood, and D.~Goldman} (2006): \enquote{HIV
  breakthroughs and risky sexual behavior,} \emph{The Quarterly Journal of
  Economics}, 121, 1063--1102.

\bibitem[\protect\citeauthoryear{McIntosh}{McIntosh}{2007}]{mcintosh2007}
\textsc{McIntosh, C.} (2007): \enquote{Has Better Health Care Contributed to
  Higher HIV Prevalence in Sub-Saharan Africa?} Working paper, University of
  California, San Diego.

\bibitem[\protect\citeauthoryear{Milgrom and Roberts}{Milgrom and
  Roberts}{1994}]{milgrom+roberts1994}
\textsc{Milgrom, P. and J.~Roberts} (1994): \enquote{Comparing equilibria,}
  \emph{The American Economic Review}, 84, 441--459.

\bibitem[\protect\citeauthoryear{O'Donnell, Stueve, San~Doval, Duran, Haber,
  Atnafou, Johnson, Grant, Murray, Juhn et~al.}{O'Donnell
  et~al.}{1999}]{odonnell+1999}
\textsc{O'Donnell, L., A.~Stueve, A.~San~Doval, R.~Duran, D.~Haber, R.~Atnafou,
  N.~Johnson, U.~Grant, H.~Murray, G.~Juhn, et~al.} (1999): \enquote{The
  effectiveness of the Reach for Health Community Youth Service learning
  program in reducing early and unprotected sex among urban middle school
  students.} \emph{American Journal of Public Health}, 89, 176--181.

\bibitem[\protect\citeauthoryear{Rowthorn and Toxvaerd}{Rowthorn and
  Toxvaerd}{2017}]{rowthorn+toxvaerd2017}
\textsc{Rowthorn, B.~R. and F.~Toxvaerd} (2017): \enquote{The optimal control
  of infectious diseases via prevention and treatment,} Discussion Paper
  DP8925.

\bibitem[\protect\citeauthoryear{Talam{\`a}s and Vohra}{Talam{\`a}s and
  Vohra}{2019}]{talamas+vohra2019}
\textsc{Talam{\`a}s, E. and R.~Vohra} (2019): \enquote{Free and perfectly safe
  but only partially effective vaccines can harm everyone,} Working paper,
  University of Pennsylvania.

\bibitem[\protect\citeauthoryear{Toxvaerd}{Toxvaerd}{2019}]{toxvaerd2019}
\textsc{Toxvaerd, F.} (2019): \enquote{Rational disinhibition and externalities
  in prevention,} \emph{International Economic Review}.

\bibitem[\protect\citeauthoryear{Underhill, Montgomery, and Operario}{Underhill
  et~al.}{2007}]{underhill+2007}
\textsc{Underhill, K., P.~Montgomery, and D.~Operario} (2007): \enquote{Sexual
  abstinence only programmes to prevent HIV infection in high income countries:
  systematic review,} \emph{British Medical Journal}, 335, 248.

\bibitem[\protect\citeauthoryear{Whitaker and Renton}{Whitaker and
  Renton}{1992}]{whitaker+renton1992}
\textsc{Whitaker, L. and A.~Renton} (1992): \enquote{A theoretical problem of
  interpreting the recently reported increase in homosexual gonorrhoea,}
  \emph{European {J}ournal of {E}pidemiology}, 8, 187--191.

\end{thebibliography}
\end{document}